\documentclass[preprint,12pt]{elsarticle}

\usepackage[english]{babel}

\usepackage[a4paper,top=2cm,bottom=2cm,left=3cm,right=3cm,marginparwidth=1.75cm]{geometry}

\usepackage{amsmath}
\usepackage{graphicx}
\usepackage[colorlinks=true, allcolors=blue]{hyperref}
\usepackage{amssymb}
\usepackage{amsthm}

\def\bx{{\mathbf{x}}}
\def\by{{\mathbf{y}}}

\def\Z{{\mathbb Z}}

\def\N{{\mathbb N}}

\def\F{{\mathbb F}}

\def\calY{{\mathcal Y}}
\def\calX{{\mathcal X}}
\newcommand\ev{{\operatorname{ev}}}

\newtheorem{theorem}{Theorem}[section]
\newtheorem{lemma}[theorem]{Lemma}

\newtheorem{corollary}[theorem]{Corollary}

\theoremstyle{definition}

\newtheorem{example}[theorem]{Example}

\theoremstyle{remark}

\newtheorem{notation}[theorem]{Notation}

\newtheorem{fact}[theorem]{Fact}

\numberwithin{equation}{section}
\begin{document}



\title{Some Patterns of Duplications in the outputs of Mersenne Twister
Pseudorandom Number Generator MT19937}
\author[label1]{Alain Schumacher}

\affiliation[label1]{organization={SICAP R\&D},
addressline={L-1272 Luxembourg-Beggen 68, rue de Bourgogne},
city={Luxembourg},
postcode={},
country={Luxembourg}}

\author[label2]{Takuji Nishimura}
\affiliation[label2]{organization={Yamagata Ujniversity},
addressline={1-4-12 Kojirakawa},
city={Yamagata},
postcode={990-8560},
country={Japan}}
\author[label3]{Makoto Matusmoto}
\affiliation[label3]{organization={AMAGAERU Institute of Free Mathematics},
addressline={2-37-6, Narita-Higashi, Suginami-ku},
city={Tokyo},
postcode={166-0015},
country={Japan}}

\begin{abstract}
The Mersenne Twister MT19937 pseudorandom number generator,
introduced by the last two authors in 1998, is still widely used.
It passes all existing statistical tests,
except for the linear complexity test, which measures
the ratio of the even-odd of the number of 1's among
specific bits (and hence should not be important for
most applications). Harase reported that MT19937 is rejected
by some birthday-spacing tests, which are rather artificially
designed.

In this paper, we report that MT19937 fails in a natural
test based on the distribution of run-lengths on which
we found an identical value in the output 32-bit integers.
The number of observations of the run-length 623 is some 40 times larger
than the expectation (and than the numbers of the observations
of 622 and 624, etc.), which implies that the corresponding p-value is
almost 0.

We mathematically analyze the phenomena, and
obtain a theorem which explains these failures.
It seems not to be a serious defect of MT19937,
because finding the defect requires astronomical efforts.
Still, the phenomena should be reported to the academic society
relating to pseudorandom number generation.
\end{abstract}
\maketitle
\section{Introduction}
A Mersenne Twister MT19937 pseudo random number generator \cite{MT}
is still widely used.
As far as the authors know, no reasonable statistical
tests reject MT19937, except for the
linear complexity test.
The test measures whether the number of 1's
in the bit-presentation of some of the output integers
is even or not, and hence it matters seldomly.
The first failure of MT19937 in other statistical tests
than the linear complexity test is reported
by Harase \cite{HARASE-MTLATTICE}.
It is found that MT19937 has linear relations among
three non-successive outputs and is
rejected by a birthday-spacing test on them.
However, the test is rather artificial,
concentrating on outputs of fixed non-successive
outputs, such as
$y_i, y_{i+396}$ and
$y_{i+623}$ $(i=0,1,\ldots)$.

In this paper, we report that a
modified version of the repetition test
introduced by
Gil, Gonnet, and Petersen \cite{REPETITION}
naturally detects
some flaws of MT19937.
The flaws are shown in Section~\ref{sec:FLAW},
which are on positive correlations between
repetitions of identical 32-bit integers.
In Section~\ref{sec:proof},
we prove that these flaws are
due to the sparseness of
a matrix $C$ appeared in the
recursion of MT19937
(this is not the case for MT19937-64 \cite{MT64}).
We
explain how the $C$ yields
some clear patterns on the repetitions.
This section is mathematically rather complicated,
and one may skip it
and go to Section~\ref{sec:REPETITION},
where we explain how these flaws are
found through experiments on
the repetition tests, i.e., by observing the
distribution of the run-lengths to find a repetition.

\section{Flaws of Mersenne Twister}
\label{sec:FLAW}
Let $\N$ denote the set of non-negative integers.
We report flaws of Mersenne Twister MT19937.
\begin{fact} Put $n=624$ and $m=397$.
Let $\by_0, \by_1, \by_2, \ldots$ be the 32-bit integer
outputs of MT19937.
Let $i$ be an arbitrary integer.
Then, the following are strongly positively correlated.
\begin{enumerate}[(1)]
\item $\by_{i+(m-1)}=\by_{i+(n-1)}$.
\item $\by_{i+2(m-1)}=\by_{i+2(n-1)}$.
\item $\by_{i+4(m-1)}=\by_{i+4(n-1)}$.
\item $\by_{i+8(m-1)}=\by_{i+8(n-1)}$.
\item $\by_{i+16(m-1)}=\by_{i+16(n-1)}$.
\item $\by_{i+32(m-1)}=\by_{i+32(n-1)}$.
\end{enumerate}
To be more precise,
under the condition ($i$) above ($i=1,2,3,4,5$),
($i+1$) occurs with the probability \emph{exactly}
$1/2, 1/4, 1/16, 1/256, 1/65536$, respectively.
These numbers are too large, since if the outputs
were truly random, each probability
should be $2^{-32}$.
One observes that, for example, a triplet consisting from
three of the six conditions appears too often.
E.g.,
the probability that (2), (3) and (4) occur at the same time is
$1/64\cdot 2^{-32}$, which is far too large than the expectation
$2^{-3\times 32}$ for truely random sequences.
\end{fact}
See Example~\ref{ex:prob} below, and for
the general formula Theorem~\ref{th:main}.
These flaws are discovered by the first author
experimentally, when
a large variant of the repetition test
introduced by Gil, Gonnet and Petersen \cite{REPETITION} is executed,
see Section~\ref{sec:REPETITION}.

\section{Theorem and Proof}
\label{sec:proof}
\begin{notation}
Let $\F_2$ denote the two element field $\{0,1\}$
with addition is specified by the exor $1+1=0$.
We identify the set of 32-bit integers with
the set of horizontal vectors
$W=\F_2^{32}$ (W for words). Let $w$ be $32$.
Consider the set of sequences of $W$ indexed with $\Z$,
namely,
$$
W^\Z:=\{(\ldots, \by_2,\by_1,\by_0,\by_{-1},\by_{-2},\ldots)
\mid \by_i\in W\}.
$$
Let $\calY=(\by_i)_{i\in \Z}$ be a sequence in $W^\Z$.
We define a delay operator
$$
D:W^\Z \to W^\Z
$$
by mapping
$$(\by_i \mid i\in \Z) \mapsto (\by_{i+1} \mid i\in \Z).$$
We denote this action from the right:
$$
\calY \mapsto \calY D.
$$
This means that $D$ shifts the components by one to the right.
Cleary, $D^{-1}$
exists, which shifts the components in the other direction.
A $(w\times w)$ matrix $M$ acts on $W^Z$ diagonally from right,
i.e.,
$$
\calY M = (\by_i M \mid i\in \Z).
$$
This action
is denoted by the same letter, thus
$$
M: W^\Z \to W^\Z.
$$
It is easy to check that $D$ and $M$ commute.
We define evaluation at $j\in \Z$ by
$$
\ev_j:W^\Z \to W, \quad \calY \to \by_j.
$$
We have
$$
\ev_j(D^k(\calY))=\by_{j+k}.
$$
\end{notation}
Since the output of MT19937 is purely periodic,
we may consider them as a sequence indexed by $\Z$.
Also, since a tempering is a (linear) bijection,
which preserves the identity relation,
we may consider
the sequences generated by the recursion \cite[(2.1)]{MT}
(i.e., before tempering)
as the output sequence of MT19937, as far as
we consider only the repetition of outputs.
The next lemma gives an equivalent condition to the
recursion formula in \cite[Section~2.1]{MT}.
\begin{lemma}
The sequences which MT19937 produces (before tempering) are
characterized as a kernel of the operator
\begin{equation}\label{eq:operator}
D^{n}-D^{m}+DB+C: W^\Z \to W^\Z,
\end{equation}
where
\begin{eqnarray*}
B&=&
\begin{pmatrix}
0 & 0 & 0 & 0 & \cdots & 0 \\
0 & 0 & 1 & 0 & \cdots & 0 \\
0 & 0 & 0 & 1 & \cdots & 0 \\
0 & 0 & 0 & 0 & \ddots & 0 \\
\vdots & \vdots & \vdots & \vdots & \vdots & \vdots \\
0 & 0 & 0 & 0 & \cdots & 1 \\
a_{w-1} & a_{w-2} & a_{w-3} & a_{w-4}& \ddots & a_0
\end{pmatrix} \\
&=&
\left(\begin{array}{c|c}
0 & 0 \cdots 0 \\ \hline
\begin{matrix}
0 \\
\vdots \\
0 \\
\\
\end{matrix} & \text{\huge{F}}
\end{array}\right),
\end{eqnarray*}
where $F$ is a companion matrix,
and
\begin{equation}
C=
\begin{pmatrix}
0 & 1 & 0 & 0 & \cdots & 0 \\
0 & 0 & 0 & 0 & \cdots & 0 \\
\vdots & \vdots & \vdots & \vdots & \vdots & \vdots \\
0 & 0 & 0 & 0 & \cdots & 0 \\
\end{pmatrix}.
\end{equation}
\end{lemma}
\begin{proof}
Let $\calX=(\bx_i)_{i\in \Z}$ be a sequence in $W^\Z$.
We compute
\begin{eqnarray*}
\ev_k((\calX)(D^{n}-D^{m}-DB-C))
&=& \ev_k((\calX)D^{n}-\calX D^{m}-\calX DB-\calX C)\\
&=& \bx_{k+n}-\bx_{k+m}-\bx_{k+1}B-\bx_k C.\\
\end{eqnarray*}
This is zero for every $k$, if and only if
the recursion \cite[(2.1)]{MT} is satisfied.
\end{proof}
Note that $C$ has only one non-zero component. This sparseness
of $C$ comes from the smallness of the remainder:
$19937 \div 32 = 623$ with remainder $1$ (the case
$w-r=1$ in \cite[Section~3.1]{MT}).
For example, such a phenomenon is not observed
for MT19937-64 \cite{MT64}.

Since $D$ is invertible, we have the following
\begin{corollary}
The space of the outputs (before tempering) of
MT19937 is the kernel of the operator
\begin{equation}\label{eq:op}
D^{n-1}-D^{m-1}+B+D^{-1}C: W^\Z \to W^Z.
\end{equation}
\end{corollary}

\begin{lemma}
As an operator on $W^\Z$,
for any non-negative integer $s$,
$$
(D^{n-1}-D^{m-1}+B+D^{-1}C)^{2^s}
=
D^{2^s(n-1)}-D^{2^s(m-1)}+(B+D^{-1}C)^{2^s}.
$$
\end{lemma}
\begin{proof}
In the case $s=1$, paying attention to $1+1=0$ and
that $D$ commutes with any other operators,
we have
\begin{eqnarray*}
& &
(D^{n-1}-D^{m-1}+B+D^{-1}C)^2 \\
&=&
(D^{n-1}-D^{m-1})^2+2((D^{n-1}-D^{m-1}))(B+D^{-1}C)+(B+D^{-1}C)^2 \\
&=&
D^{2(n-1)}-D^{2(m-1)}+(B+D^{-1}C)^2.
\end{eqnarray*}
The cases for general $s$ follow by induction.
\end{proof}
\begin{corollary}\label{cor:repetition}
Let $\calX=(\bx_i \mid i \in \Z)$ be a sequence generated by MT19937.
Let $s$ be a non-negative integer.
Then,
$$
\bx_{i+2^s(m-1)}=\bx_{i+2^s(n-1)}
$$
holds if and only if
$$
\ev_i(\calX(B+D^{-1}C)^{2^s})=0.
$$
\end{corollary}
\begin{proof}
Since $\calX$ is generated by MT19937,
$$
\calX(D^{n-1}-D^{m-1}+B+D^{-1}C)^{2^s}=0.
$$
It follows that
$$
\calX(D^{2^s(n-1)}-D^{2^s(m-1)}+(B+D^{-1}C)^{2^s})=0.
$$
By $\ev_i$, this yields the desired statement.
\end{proof}
We compute the above for some small $s$.
It is easy to check that $C^2=0$,
and
$$
CB=
\begin{pmatrix}
0 & 0 & 1 & 0 & \cdots & 0 \\
0 & 0 & 0 & 0 & \cdots & 0 \\
\vdots & \vdots & \vdots & \vdots & \vdots & \vdots \\
0 & 0 & 0 & 0 & \cdots & 0
\end{pmatrix},
$$
and
$$
BC
=
\begin{pmatrix}
0 & 0 & 0 & 0 & \cdots & 0 \\
0 & 0 & 0 & 0 & \cdots & 0 \\
\vdots & \vdots & \vdots & \vdots & \vdots & \vdots \\
0 & 0 & 0 & 0 & \cdots & 0 \\
0 & 1 & 0 & 0 & \cdots & 0
\end{pmatrix},
$$
both having only one non-zero component.
\begin{lemma}\label{lem:B-power}
For $s<w$, we have
$$
B^s=
\left(
\begin{array}{c|c}
0 & 0 \cdots 0\\ \hline
\begin{matrix}
0 \\
\vdots \\
0 \\
\\
\\
\vdots \\
*
\end{matrix} & \text{\huge F}^{\ s}
\end{array}
\right)
$$
where in the first column, $1$ is at the $s$-th row
from the bottom (i.e. the $(w-s+1)$-st row from the top).
We have
$$
CB^sC=0
$$
for $s\leq w-2$.
\end{lemma}
\begin{proof}
The first statement follows from the induction on $s$ and the property
of a companion matrix $F$. For the last statement, $CXC$
has only one possibly non-zero component, whose value is
the $(2,1)$-component of $X$, which is zero for $B^s$ with $s\leq w-2$.
\end{proof}
\begin{corollary}\label{cor:BCB}
For $0\leq t\leq w-2$, $CB^t$ is a matrix whose components are
all zero, except the one at the $(1,t+2)$-component.
For $1\leq s \leq w-2$, $B^sCB^t$
is a matrix whose columns are zero,
except the $(t+2)$-nd column, which
has $w-s$ zeroes at the top, and the $(w-s+1)$-st
component is one.
\end{corollary}
\begin{proof}
For a horizontal vector $\by=(y_1,\ldots,y_w)$ with $y_w=0$,
$$
\by B =(0,0,y_2,y_3,\ldots,y_{w-1}),
$$
namely,
obtained from $\by$ by replacing
$y_1$ by $0$ and shift right. Because each row of $C$
has $w-2$ zeroes at the right,
$CB^t$ is obtained by this shifting for $t\leq w-2$,
which proves the first statement.
By Lemma~\ref{lem:B-power}, we know the first column of
$B^s$, and $B^sC$ has a unique nonzero column at the second row,
which is identical with the first column of $B^s$.
The form of $B^sCB^t$ follows.
\end{proof}

\begin{lemma}\label{lem:Q}
For $k\leq w-2$, we have
\begin{equation}\label{eq:BC-power}
(B+D^{-1}C)^{k}
=
B^k + D^{-1}\sum_{i=0}^{k-1}B^iCB^{k-i-1}.
\end{equation}
We define
$$
Q_{k}:=\sum_{i=0}^{k-1}B^iCB^{k-i-1}.
$$
Then, for $2^s \leq w-2$,
\begin{eqnarray}\label{eq:BC2-power}
Q_{2^s}&=&
Q_{2^{s-1}}B^{2^{s-1}}
+
B^{2^{s-1}}Q_{2^{s-1}}
\end{eqnarray}
holds.
\end{lemma}
\begin{proof}
The left hand side of (\ref{eq:BC-power})
is the sum of all the possible $2^k$
monomials consisting of $k$ of $B$ or $D^{-1}C$.
By Lemma~\ref{lem:B-power},
the terms with two $C$'s are zero. Thus (\ref{eq:BC-power})
follows. By the case division of the place of $C$,
(\ref{eq:BC2-power}) follows.
\end{proof}

\begin{lemma}\label{lem:rank-w}
For any $s\geq 0$, the rank of
$
\begin{pmatrix}
B^s \\
Q_s
\end{pmatrix}
$
is $w$.
\end{lemma}
\begin{proof}
We remark that $B+C$ is a companion matrix, and since
$a_{w-1}=1$, it is invertible. Thus, $(B+C)^s=B^s+Q_s$
is invertible and has rank $w$. Hence
$
\begin{pmatrix}
B^s \\
Q_s
\end{pmatrix}
$
has rank no less than $w$, and being $2w\times w$ matrix,
it has rank $w$.
\end{proof}

\begin{theorem}\label{th:main}
Assume that the initialization of MT19937 is done
uniformly (including zero).
Let $0\leq s<t$ be integers with $2^t\leq w-2$.
The probability that
\begin{equation}\label{eq:event-k}
\bx_{i+2^k(m-1)}=\bx_{i+2^k(n-1)}
\end{equation}
holds for all $k$, $s\leq k \leq t$,
is
$$
2^{-w}\cdot 2^{-(2^{t}-2^{s})}.
$$
(This is much higher than
$2^{-w(t-s+1)}$ for a true random sequence,
since $2^t\leq w-2$.)
\end{theorem}
For example, we choose $s=0$ and $t=1$.
Then, the probability that
$$
\bx_{i+(m-1)}=\bx_{i+(n-1)} \mbox{ and }
\bx_{i+2(m-1)}=\bx_{i+2(n-1)}
$$
occur is $1/2\cdot 2^{-w}$, while
for true random numbers, the probability is $2^{-2w}$.
\begin{proof}
By Corollary~\ref{cor:repetition},
(\ref{eq:event-k}) is equivalent to
$$\ev_i(\calX(B+D^{-1}C)^{2^k})=0,$$
and by Lemma~\ref{lem:Q} equivalent to
$$
\ev_i(\calX(B^{2^k}+D^{-1}Q_{2^k}))=0,
$$
which is
\begin{equation}\label{eq:on-2k}
\bx_{i}B^{2^k}+\bx_{i-1}Q_{2^k}=0.
\end{equation}
The conditions for all $k$
satisfying $s\leq k \leq t$ can be described by the
product of a vector and a matrix
\begin{equation}\label{eq:rank}
(\bx_i, \bx_{i-1})
\begin{pmatrix}
B^{2^s} & B^{2^{s+1}} &\cdots & B^{2^{t-1}} &B^{2^t} \\
Q_{2^s} & Q_{2^{s+1}} &\cdots & Q_{2^{t-1}} &Q_{2^t}
\end{pmatrix} = 0.
\end{equation}
Because of the random choice of the initial seed,
$(\bx_i, \bx_{i-1})$ is uniformly random
(note that this property is called the 2-dimensional
equidistribution, while MT19937 is known to be 623-dimensionally
equidistributed), and
the probability that this equality holds is
$2^{-r}$, where $r$ is the rank of the matrix in (\ref{eq:rank}).

By multiplying $B^{2^{t-1}}$ from the right to the
second (from the right end) row and subtracting it from the right
most row, we have a matrix with the same rank
$$
\begin{pmatrix}
B^{2^s} & B^{2^{s+1}} &\cdots & B^{2^{t-1}} &0 \\
Q_{2^s} & Q_{2^{s+1}} &\cdots & Q_{2^{t-1}} & B^{2^{t-1}}Q_{2^{t-1}}
\end{pmatrix},
$$
where we use (\ref{eq:BC2-power}) for the right-bottom corner.
Then, we multiply $B^{2^{t-2}}$ to the third row from the right,
and subtract it from the second row from the right. By iteration, we have
a matrix with the same rank
\begin{equation}\label{eq:want-rank}
\begin{pmatrix}
B^{2^s} & 0&\cdots & 0&0 \\
Q_{2^s} & B^{2^s}Q_{2^s} &\cdots & B^{2^{t-2}}Q_{2^{t-2}} & B^{2^{t-1}}Q_{2^{t-1}}
\end{pmatrix}.
\end{equation}
Here we have
$$
B^{2^k}Q_{2^k}
=B^{2^k}\sum_{i=0}^{2^k-1}B^iCB^{2^k-i-1}
=\sum_{i=0}^{2^k-1}B^{2^k+i}CB^{2^k-i-1}.
$$
By Corollary~\ref{cor:BCB},
$B^{2^k+i}CB^{2^k-i-1}$ has a unique nonzero column
as the $(2^k-i-1+2)$-nd column
with top $w-(2^k+i)$ components being zeroes
and the $(w-(2^k+i)+1)$-st component is one.
The range of $i$ is $0\leq i \leq 2^{k}-1$.
This means that all these columns for $k$,
$s\leq k\leq t$, and $0\leq i \leq 2^{k}-1$
are linearly independent.
Their number is
$$
2^s+2^{s+1}+\cdots+2^{t-1}=2^t-2^s.
$$
This means that we have a
$2w\times (w+2^t-2^s)$-matrix with the same rank as (\ref{eq:want-rank})
\begin{equation}\label{eq:rank-final}
\begin{pmatrix}
B^{2^s} & 0\\
Q_{2^s} & G
\end{pmatrix},
\end{equation}
where $G$ consists of the above $2^t-2^s$ columns.
We show that the columns of (\ref{eq:rank-final})
are independent. Let $b_1,\ldots,b_{w},c_{1},\ldots,c_{2^t-2^s}$
be elements of $\F_2$, and the linear combination
of the columns with these coefficients is zero.
Then, since $G$ has independent columns, it
follows that $c_i$ for $i=1,\ldots,2^t-2^s$ are zeroes.
Then, by Lemma~\ref{lem:rank-w}, $b_i$ for $i=1,\ldots,w$ are zeroes,
hence the columns of (\ref{eq:rank-final}) are linearly independent.
Thus the matrix (\ref{eq:rank-final}) has
rank $w+2^t-2^s$.
Thus, the probability that all equalities in
(\ref{eq:on-2k}) for $s\leq k \leq t$ hold is
$2^{-w-(2^t-2^s)}$, which proves the theorem.
\end{proof}

\begin{example}\label{ex:prob}
Put $n=624$ and $m=397$.
Let $\by_0, \by_1, \by_2, \ldots$ be the 32-bit integer
outputs of MT19937.
Let $i$ be an arbitrary integer.
Then, the following are strongly positively correlated.
\begin{enumerate}[(1)]
\item $\by_{i+(m-1)}=\by_{i+(n-1)}$.
\item $\by_{i+2(m-1)}=\by_{i+2(n-1)}$.
\item $\by_{i+4(m-1)}=\by_{i+4(n-1)}$.
\item $\by_{i+8(m-1)}=\by_{i+8(n-1)}$.
\item $\by_{i+16(m-1)}=\by_{i+16(n-1)}$.
\item $\by_{i+32(m-1)}=\by_{i+32(n-1)}$.
\end{enumerate}
For example, the probability that (1) and (2)
occur is the case $s=0$ and $t=1$, hence
$2^{-w}2^{-(2^t-2^s)}=1/2\cdot 2^{-w}$, while
the truly random sequence has the probability
$2^{-w}\cdot 2^{-w}$.

The probability that (4) and (5) occur
is the case $s=3$ and $t=4$, hence
$2^{-w}2^{-(2^4-2^3)}=1/256\cdot 2^{-w}$.

The probability that (5) and (6) occur is
$2^{-w}2^{-(2^5-2^4)}=1/65536\cdot 2^{-w}$.

For triples, for example,
the probability that (2), (3) and (4) occur
is the case $s=1$ and $t=3$, hence is $2^{-w}2^{-(8-2)}=1/64\cdot 2^{-w}$,
while the probability for a truly random sequence is
$2^{-3w}$.
\end{example}

\section{Modified repetition tests found the patterns}
\label{sec:REPETITION}
The original repetition test
counts the run-length for
observing the identical
32-bit occurring at two distinct places.
Namely, starting from the
output $\by_0$ of MT19937,
we memorize consecutive outputs,
until we found $r$ such that $\by_{r-d}=\by_{r}$.
This $r$ is called the run-length
for the repetition.

After finding such an $r$, we do not
initialize MT19937, just continue
to find the next repetition.
In the original test, one set is to iterate this 100 times,
and the average of $r$'s is taken and compared with
its theoretical distribution.
(Three sets are repeated.)

The first author iterated this 100 billion times,
and instead of taking the average,
he observed the number of occurrences
of $r$, for $r=2,3,\ldots, 2100000$.
The probability that $r>2100000$
occurs is negligibly small
($0.000\cdots$ with 222 zeros).
Then, the number of the occurrences of the case 
$r=623$ is extremely high
(more than 40 times larger than the expectation, as stated above).
Analogous phenomena are observed at
$r=1246, 2492, \ldots$.

Then, the second and the third authors
analyzed the phenomena, and
found Theorem~\ref{th:main},
which explains these phenomena.
Probability that
$r=623$
occurs is very small for
a true random 32-bit integer sequence.
For the case of MT19937,
after observing
the first repetition, the continuous
search for the next repetition
is affected by the
strong correlation
among (1) and (2) in Theorem~\ref{th:main}.
Namely, it is often the case that
(2) is preceded by (1), and (1) is
detected in the previous run.
A part of Theorem
implies that (2) is very often after (1)
(i.e., with probability 1/2, whereas it should be $2^{-32}$
if the sequence is truly random),
which means that we should observe the run-length
$n-1=623$ quite often
(after finding (1), we continue to search for the next,
which implies that we observe (2) very often,
where the run-length is $n-1=623$).

The first author noticed that the cases
(2) after (1), (3) after (2), and
(3) after (2) after (1), etc.,
are extremely frequent.
The main theorem is proved to
explain these phenomena.
The exact figures of the cases encountered are available 
for download at: 
\cite{SICAP-DOWNLOADS}.



\bibliographystyle{elsarticle-num}
\bibliography{sfmt-kanren}

@article{MT,
author = "M. Matsumoto and T. Nishimura",
title = "Mersenne Twister: A 623-dimensionally equidistributed uniform pseudorandom number generator",
journal = "ACM Trans. on Modeling and Computer Simulation",
volume = "8",
year = "1998",
pages = "3-30",
number = "1",
month = jan,
note = "\url{http://www.math.sci.hiroshima-u.ac.jp/~m-mat/MT/emt.html}"
}

@article{REPETITION,
author = "M. Gil and G. H. Gonnet and W. P. Petersen",
title = "A Repetition Test for Pseudo-Random Number
Generators",
journal = "Monte Carlo Methods and Appl.",
volume = "12",
year = "2006",
pages = "385-393",
number = "5-7",
}

@article{MT64,
author = "T. Nishimura",
title = "Tables of 64-bit Mersenne Twisters",
journal = "ACM Trans. on Modeling and Computer Simulation",
volume = "10",
year = "2000",
pages = "348-357",
number = "4",
month = oct
}

@article{HARASE-MTLATTICE,
author = "S. Harase",
title = "On the {F2}-linear relations of {Mersenne} Twister pseudorandom number generators",
journal = "Mathematics and Computers in Simulation", 
volume = "100",
year = "2014",
pages = "103-113",
note = "arXiv:1301.5435",
}

@url{SICAP-DOWNLOADS,
author = "A. Schumacher",
title = "SICAP Downloads Page",
note = "home page",
url = "https://sicap.lu/Fhtml/Downloads.html",
}





\end{document}